\documentclass[11pt]{article}
\usepackage[utf8]{inputenc}

\usepackage{a4wide}

\usepackage{graphicx,color}
\usepackage{tikz}
\usetikzlibrary{arrows}

\usepackage{textcomp}
\usepackage{latexsym}
\usepackage{amsmath}
\usepackage{amsfonts}
\usepackage{amssymb}
\usepackage{amsthm}
\usepackage{hyperref}
 
\newtheorem{theorem}{Theorem}
\newtheorem{theorem*}{Theorem}
\newtheorem{corollary}[theorem]{Corollary}
\newtheorem{lemma}[theorem]{Lemma}

\newtheorem{proposition}[theorem]{Proposition}
\newtheorem{definition}[theorem]{Definition}
\newtheorem{claim}[theorem]{Claim}

\newtheorem{remark}[theorem]{Remark}
\newtheorem{conjecture}[theorem]{Conjecture}

\newcommand{\F}{\mathbb{F}}

\newcommand{\C}{\mathbb{C}}

\renewcommand{\angle}[1]{\mathopen{\langle} #1\mathclose{\rangle}}

\newcommand{\VP}{\mathrm{VP}}
\newcommand{\VNP}{\mathrm{VNP}}

\renewcommand{\P}{\mathrm{P}}
\newcommand{\NP}{\mathrm{NP}}

\newcommand{\vpnc}{\mathrm{VP}_{nc}}
\newcommand{\vnpnc}{\mathrm{VNP}_{nc}}

\newcommand{\vsk}{\mathrm{VSKEW}_{nc}}

\newcommand{\vbpnc}{\mathrm{VBP}_{nc}}

\newcommand{\vncnc}{\mathrm{VNC}_{nc}}

\newcommand{\abp}{\le_{abp}}

\newcommand{\pal}{\mathrm{PAL}}

\newcommand{\poly}{\mathrm{poly}}

\newcommand{\DET}{\mathrm{DET}}

\newcommand{\PER}{\mathrm{PER}}

\newcommand{\PAL}{\mathrm{PAL}}

\title{Noncommutative Valiant's Classes: Structure and Complete Problems}

\author{V. Arvind\thanks{Institute of Mathematical Sciences, Chennai,
    India, \texttt{email: arvind@imsc.res.in}} \and Pushkar S
  Joglekar\thanks{Vishwakarma Institute of Technology, Pune, India,
    \texttt{email: joglekar.pushkar@gmail.com}} \and S. Raja\thanks{Institute of Mathematical Sciences, Chennai, India,
    \texttt{email: rajas@imsc.res.in}}}

\begin{document}

\maketitle

\begin{abstract}
In this paper we explore the noncommutative analogues, $\vpnc$ and
$\vnpnc$, of Valiant's algebraic complexity classes and show some
striking connections to classical formal language theory. Our main
results are the following:

\begin{itemize}
\item We show that Dyck polynomials (defined from the Dyck languages
  of formal language theory) are complete for the class $\vpnc$ under
  $\abp$ reductions. Likewise, it turns out that $\pal$ (Palindrome
  polynomials defined from palindromes) are complete for the class
  $\vsk$ (defined by polynomial-size skew circuits) under $\abp$
  reductions. The proof of these results is by suitably adapting the
  classical Chomsky-Sch\"{u}tzenberger theorem showing that Dyck
  languages are the hardest CFLs.

\item Next, we consider the class $\vnpnc$. It is known~\cite{HWY10a}
  that, assuming the sum-of-squares conjecture, the noncommutative
  polynomial $\sum_{w\in\{x_0,x_1\}^n}ww$ requires exponential size
  circuits. We unconditionally show that $\sum_{w\in\{x_0,x_1\}^n}ww$
  is not $\vnpnc$-complete under the projection reducibility. As a
  consequence, assuming the sum-of-squares conjecture, we exhibit a
  strictly infinite hierarchy of p-families under projections inside
  $\vnpnc$ (analogous to Ladner's theorem~\cite{Ladner75}). In the
  final section we discuss some new $\vnpnc$-complete problems under
  $\abp$-reductions.

\item Inside $\vpnc$ too we show there is a strict hierarchy of
  p-families (based on the nesting depth of Dyck polynomials) under
  the $\abp$ reducibility.
\end{itemize}

\end{abstract}

\section{Introduction}

Proving superpolynomial size lower bounds for arithmetic circuits that
compute the permanent polynomial $\PER_n$ is a central open problem in
computational complexity theory. This problem has a rich history in
the field, starting with the work of Strassen on matrix multiplication
\cite{str}.

In the late 1970's, Valiant, in his seminal work \cite{Val79}, defined
the arithmetic analogues of P and NP: namely VP and VNP. Informally,
VP consists of multivariate (commutative) polynomials over a field
$\F$ that have polynomial size circuits. The class VNP, which
corresponds to NP (in fact $\#$P to be precise) has a more technical
definition which we will give later. Valiant showed that $\PER_n$ is
VNP-complete w.r.t projection reductions. Thus, $\VP \neq \VNP$ iff
$\PER_n$ requires superpolynomial in $n$ size arithmetic circuits.

In 1990 paper \cite{N91}, Nisan explored the same question
for \emph{noncommutative} polynomials. The noncommutative polynomial ring
$\F\angle{x_1,\ldots,x_n}$ consists of $\F$-linear combinations of
words (we call them monomials) over the alphabet $X=
\{x_1,\ldots,x_n\}$.

We can analogously define noncommutative arithmetic circuits for
polynomials in $\F\angle{X}$ where the inputs to multiplication gates are
ordered from left to right. A natural definition of the noncommutative
$\PER_n$ is 
\[
\PER_n=\sum_{\sigma \in S_n}x_{1,\sigma(1)}x_{2,\sigma(2)}\ldots x_{n,\sigma(n)}
\]
over $X=\{x_{ij}\}_{1 \leq i,j\leq n}$. 

Can we show that $\PER_n$ requires superpolynomial size noncommutative
arithmetic circuits? One would expect this problem to be easier than
in the commutative setting. Indeed, for the model of noncommutative
algebraic branching programs (ABPs), Nisan \cite{N91} showed
exponential lower bounds for $\PER_n$ (and even the determinant
polynomial $\DET_n$). Unlike in the commutative world, where ABPs are
nearly as powerful as arithmetic circuits, in the noncommutative
setting, Nisan \cite{N91} could show an exponential separation between
noncommutative circuits and noncommutative ABPs. However, showing that
$\PER_n$ requires superpolynomial size noncommutative arithmetic
circuits remains an open problem.

Analogous to $\VP$ and $\VNP$, the classes $\vpnc$ and $\vnpnc$ can be
defined, as has been done by Hrubes et al \cite{HWY10b}. In
\cite{HWY10b} they have shown that $\PER_n$ is $\vnpnc$-complete w.r.t
projections (the Valiant \cite{Val79} notion which allows variables or
scalars to be substituted for variables).

The purpose of our paper is to study the structure of the classes
$\vpnc$ and $\vnpnc$ and its connections to formal language
classes. Our main results show a rich structure within $\vnpnc$ and
$\vpnc$ which nicely corresponds to properties of regular languages
and context-free languages.

\subsection{Main results of the paper}

We begin with some formal definitions needed to summarize our main results.
Detailed definitions are given in the next section.

\begin{definition}\label{pfamily}
A sequence $f= (f_n)$ of noncommutative multivariate polynomials over
a field $\F$ is called a \emph{polynomial family} (abbreviated as
p-family henceforth) if both the number of variables in $f_n$ and the
degree of $f_n$ are bounded by $n^c$ for some constant $c>0$.
\end{definition}

\begin{definition}\hfill{~}
\begin{itemize}
 \item The class $\vbpnc$ consists of p-families $f= (f_n)$
   such that each $f_n$ has an algebraic branching program (ABP) of
   size bounded by $n^b$ for some $b>0$ depending on $f$.
\item The class $\vpnc$ consists of p-families $f=(f_n)$ such
  that each $f_n$ has an arithmetic circuit of size bounded by $n^b$
  for some $b>0$ depending on $f$.
\item A p-family $f=(f_n)$ is in the class $\vnpnc$ if there exists a
  p-family $g=(g_n) \in \vpnc$ such that for some polynomial $p(n)$ 
\[  
f_n(x_1,\ldots,x_{q(n)})=\sum_{y_1,\ldots,y_r \in
  \{0,1\}}g_{p(n)}(x_1,\ldots,x_m,y_1,\ldots,y_r).
\]

where $q(n)$ is number of variables in $f_n$.
\end{itemize}
\end{definition}

We note that the class $\vbpnc$ is defined through polynomial-size
algebraic branching programs (ABPs) which intuitively correspond to
polynomial sized finite automata. In fact noncommutative ABPs are also
studied in the literature as multiplicity automata \cite{Beim}, and
Nisan's rank lower bound argument \cite{N91} is related to the rank of
Hankel matrices in formal language theory \cite{berstel}. Moreover, 
arithmetic circuits can be seen as acyclic context-free grammars (CFGs) (if the coefficients come from the
Boolean ring instead of a field).

It turns out, as we will see in this paper, the analogy goes further
and shows up in the internal structure of $\vnpnc$ and $\vpnc$.

\begin{enumerate}
 \item We prove that the Dyck polynomials are complete for $\vpnc$
   w.r.t $\leq_{abp}$ reductions. The proof is really an arithmetized
   version of the Chomsky-Sch\"utzenberger theorem \cite{cs63} showing
   that the Dyck languages are the hardest CFLs.
\item On the same lines we show that the Palindrome polynomials
  $PAL_n=\sum_{w \in \{x_0,x_1\}^n}w.w^R$ are complete for $\vsk$,
  again by adapting the proof of the Chomsky-Sch\"utzenberger theorem.
\item Within $\vpnc$ we obtain a proper hierarchy w.r.t
  $\leq_{abp}$-reductions corresponding to the Dyck polynomials of
  bounded nesting depth. This roughly corresponds to the
  noncommutative VNC hierarchy within $\vpnc$.
\item We examine the structure within $\vnpnc$ assuming the
  sum-of-squares conjecture, under which Hrubes et al \cite{HWY10a}
  have shown that the p-family $\{ ID_d=\sum_{w \in
      \{x_0,x_1\}^d}w.w\}_d \notin \vpnc$. We prove the following results about
    $\vnpnc\setminus \vpnc$.
\begin{enumerate}
\item We prove a transfer theorem which essentially shows that if $f$
  is a $\vnpnc$-complete p-family under projections then an
  appropriately defined commutative version $f^{(c)}$ of $f$ is
  complete under projections for the commutative VNP class.

 \item Assuming the sum-of-squares conjecture we show that the
   polynomial $ID_d$ is neither in $\vpnc$ nor $\vnpnc$-complete w.r.t
   $\leq_{proj}$ reductions. This is analogous to Ladner's well-known
   theorem \cite{Ladner75}.

\item It also turns out that under the sum-of-squares conjecture we
  have an infinite hierarchy w.r.t $\leq_{proj}$ reductions between
  $\vpnc$ and $\vnpnc$ and also incomparable p-families.
\end{enumerate}
\end{enumerate}

Table 1 summarizes the results in this paper.

\begin{table}[!htbp]\label{tabl}
\centering
\begin{tabular}{|c|c|c|}
  \hline
  \textbf{P-family} & \textbf{Complexity Result} & \textbf{Remarks}\\
  \hline
  \hline
 $D_k, k \geq 2$ & \begin{tabular}[t]{ll}
              - $\vpnc$-Complete (Theorem~\ref{thm:vpcomplete})\\
              - $\vsk$-hard (Theorem~\ref{pal-to-d2})
            \end{tabular} 
          & w.r.t\ $\abp$-reductions\\
\hline
  $\pal_d$ & $\vsk$-Complete (Theorem~\ref{pal-vsk}) & w.r.t.\ $\abp$-reductions\\
  \hline
 $ID_d$ &   \begin{tabular}[t]{ll}
               -not $\vnpnc$-Complete (Theorem~\ref{thm:notcomplete})\\
               -not $\vpnc$-hard (Theorem~\ref{id-not-hard})\\
               -not in $\vpnc$ \cite{HWY10a}
             \end{tabular}
       &     \begin{tabular}[t]{ll}
               $\leq_{proj},\leq_{iproj}$-reductions\\
               $\leq_{proj},\leq_{iproj}$-reductions\\
               assuming $SOS_k$ conjecture
             \end{tabular}\\
  \hline
 $\PER^{*,\chi}$ & $\vnpnc$-Complete (Theorem~\ref{thm:vnpc_gen})
       &     \begin{tabular}[t]{ll}
               w.r.t.\ $\abp$-reductions
             \end{tabular}\\

  \hline
 $ID^{*}_{n}$ & $\vnpnc$-Complete (Theorem~\ref{thm:more_id*})
       &     \begin{tabular}[t]{ll}
              w.r.t.\ $\abp$-reductions
             \end{tabular}\\

  \hline
\end{tabular}
 \caption{Summary of Results}
\end{table}

\section{Preliminaries}

A \emph{noncommutative arithmetic circuit} $C$ over a field $\F$ is a
directed acyclic graph such that each in-degree $0$ node of the graph
is labelled with an element from $X\cup \F$, where $X$ is the set of
indeterminates. Each internal node has fanin two and is labeled by
either ($+$) or ($\times$) -- meaning a $+$ or $\times$ gate,
respectively. Furthermore, each $\times$ gate has a designated left
child and a designated right child. Each gate of the circuit
inductively computes a polynomial in $\F\angle{X}$: the polynomials
computed at the input nodes are the labels; the polynomial computed at
a $+$ gate (resp. $\times$ gate) is the sum (resp.  product in
left-to-right order) of the polynomials computed at its children. The
circuit $C$ computes the polynomial at the designated output node.

A \emph{noncommutative algebraic branching program} ABP (\cite{N91},
\cite{raz05PIT}) is a layered directed acyclic graph (DAG) with one
in-degree zero vertex $s$ called the \emph{source}, and one out-degree
zero vertex $t$, called the \emph{sink}. The vertices of the DAG are
partitioned into layers $0,1,\ldots,d$, and edges go only from level
$i$ to level $i + 1$ for each $i$. The source is the only vertex at
level $0$ and the sink is the only vertex at level $d$. Each edge is
labeled with a linear form in the variables $X$.  The size of the ABP
is the number of vertices.

For any $s$-to-$t$ directed path $\gamma=e_1,e_2,\ldots,e_d$, where
$e_i$ is the edge from level $i-1$ to level $i$, let $\ell_i$ denote
the linear form labeling edge $e_i$. Let
$f_\gamma=\ell_1\cdot\ell_2\cdots\ell_d$ be the product of the linear
forms in that order. Then the ABP computes the degree $d$ polynomial
$f\in\F\angle{X}$ defined as
\[
f = \sum_{\gamma\in\mathcal{P}} f_\gamma,
\]
where $\mathcal{P}$ is the set of all directed paths from $s$ to $t$.

\subsection{Polynomials}

We now define some p-families that are important for the paper.\\

\noindent\textbf{Identity Polynomials:}\\

We define the p-family $ID=(ID_n)$ which corresponds to the familiar
context-sensitive language $\{ww\mid w\in\Sigma^*\}$. 
\[
ID_n=\sum_{w \in \{x_0,x_1\}^n}ww.
\]

We will also consider some variants of this p-family in the paper.\\

\noindent\textbf{Palindrome Polynomials:}\\

The p-family $\PAL=(\PAL_n)$ corresponds to the context-free language
of palindromes.
\[
\PAL_n=\sum_{w \in \{x_0,x_1\}^n}w.w^R.
\]

\noindent\textbf{Dyck Polynomials:}\\

Let $X_i=\{(_1,)_1,...,(_i,)_i\}$ for a fixed $i \in \mathbb{N}$.  We
define the polynomial $D_{i,n}$ over the variable set $X_i$ to be sum
of all strings in $X^{2n}_i$ which are well balanced (for all the $i$
bracket types). The $D_{i,n}$ are Dyck polynomials of degree $2n$ over
$i$ different parenthesis. The corresponding p-family is denoted
$D_i=(D_{i,n})$. 



\section{The Reducibilities}

In the paper we consider three different notions of reducibility for
our completeness results and for exploring the structure of the
classes $\vnpnc,\vpnc$ and $\vsk$.

\subsection*{The projection reducibility}

The projection is essentially Valiant's notion of reduction for which
he showed $\VNP$-completeness for $\PER_n$ and other p-families in his
seminal work \cite{Val79}. Let $f=(f_n)$ and $g=(g_n)$ be two
noncommutative p-families over a field $\F$, where $\forall n$ $f_n
\in \F\angle{X_n}$ and $g_n \in \F\angle{Y_n}$. We say
$f\leq_{proj}g$ if there are a polynomial $p(n)$ and a substitution
map $\phi:Y_{p(n)}\rightarrow X_{n}\cup \F$ such that $\forall n$
$f(X_n)=g(\phi(Y_{p(n)}))$.

As shown in \cite{HWY10b} by using Valiant's original proof, the
noncommutative $\PER_n$ p-family is $\vnpnc$-complete for
$\leq_{proj}$-reducibility. 

\subsection*{The indexed-projection reducibility}

The indexed-projection is specific to the noncommutative setting. We
say $f\leq_{iproj}g$ for p-families $f=(f_n)$ and $g=(g_n)$, where
$deg(f_n)=d_n$, $deg(g_n)=d'_n$, $f_n \in \F\angle{X_n}$, and $g_n \in
\F\angle{Y_n}$, if there are a polynomial $p(n)$ and indexed projection
map
\[
\phi:[d'_{p(n)}] \times Y_{p(n)}\rightarrow X_n \cup \F,
\]
such that on substituting $\phi(i,y)$ for variable $y\in Y_{p(n)}$ occurring
in the $i^{th}$ position in a monomial of $g_{p(n)}$ we get polynomial $f_n$. 

Clearly, $\leq_{iproj}$ is more powerful than $\leq_{proj}$ and we
will show separations in this section.

\subsection*{The abp-reducibility}

The $\leq_{abp}$ reducibility is the most general notion that we will
consider. It essentially amounts to matrix substitutions for
variables, where the matrices have scalar or variable (we allow even
constant-degree monomial) entries. In terms of complexity classes we
have: $\vbpnc\subsetneq \vsk\subsetneq \vpnc \subseteq \vnpnc$. And
$\leq_{abp}$-reductions correspond to the computational power of the
class $\vbpnc$.

Formally, let $f_n \in \F\angle{X_n}$ and $g_n \in \F\angle{Y_n}$ as
before. We say $f\leq_{abp}g$ if there are polynomials $p(n),q(n)$ and
the substitution map $\phi:Y_{p(n)}\rightarrow M_{q(n)}(X_n\cup \F)$
where $M_{q(n)}(X_n\cup \F)$ stands for $q(n)\times q(n)$ matrices
with entries from $X_n \cup \F$, with the property that $f(X_n)$ is
the $(1,q(n))$-th entry of $g(\phi(Y_{p(n)}))$.

\begin{proposition}\label{transitive-abp}
Let $f,g,h \in \mathbb{F}\langle X \rangle$ such that $f \leq_{abp} g$
and $g \leq_{abp} h$ then $f \leq_{abp} h$.
\end{proposition}

\begin{proposition}
Let $f,g \in \mathbb{F}\langle X \rangle$ and $f \leq_{abp} g$. Then
if $g$ has polynomial size ABP or a noncommutative arithmetic circuit
or a noncommutative skew circuit then $f$ has polynomial size ABP, a
noncommutative arithmetic circuit, a noncommutative skew circuit
respectively.
\end{proposition}

\subsection{Hadamard product of polynomials}

We describe ideas from \cite{AJS09} that are useful for the present
paper in connection with showing $\leq_{abp}$ reductions between
p-families. Consider an ABP $P$ computing a noncommutative polynomial
$g\in\F\angle{X}$. Suppose the ABP $P$ has $q$
nodes with source $s$ and and sink $t$.

For each variable $x\in X$ we define a $q\times q$ matrix $M_x$, whose
$(i,j)^{th}$ entry $M_x(i,j)$ is the coefficient of variable $x$ in
the linear form labeling the directed edge $(i,j)$ in the ABP
$P$.\footnote{If $(i,j)$ is not an edge in the ABP then the
  coefficient of $x$ is taken as $0$.}

Consider a degree $d$ polynomial $f\in\F\angle{X}$, where
$X=\{x_1,\cdots,x_n\}$. For each monomial $w= x_{j_1}\cdots x_{j_k}$
we define the corresponding matrix product $M_w=M_{x_{j_1}}\cdots
M_{x_{j_k}}$. When each indeterminate $x\in X$ is substituted by the
corresponding matrix $M_x$ then the polynomial $f\in\F\angle{X}$
evaluates to the matrix
\[
\sum_{f(w)\ne 0} f(w)M_w,
\]
where $f(w)$ is the coefficient of monomial $w$ in the polynomial $f$.


\begin{theorem}{\rm\cite{AJS09}}
 Let $C$ be a noncommutative arithmetic circuit computing a polynomial
 $f\in \F\langle{x_1, x_2, \ldots, x_n}\rangle$. Let $P$ be an ABP (with $q$ nodes, source node
 $s$ and sink node $t$) computing a polynomial $g\in\F\langle{x_1, x_2, \ldots, x_n}\rangle$.
 Then the $(s,t)^{th}$ entry of the matrix
 $f(M_{x_1},M_{x_2},\ldots,M_{x_n})$ is the polynomial 
\[
\sum_{w}f(w)g(w)w.
\]
where $f(w),g(w)$ are coefficients of monomial $w$ in $f$ and $g$ respectively.
Hence there is a circuit of size polynomial in $n$, size of $C$ and
size of $P$ that computes the Hadamard product polynomial
$\sum_{w}f(w)g(w)w$.
\end{theorem}

\begin{remark}
A specific case of interest is when the ABP $P$ is a deterministic
finite automaton with start state $s$ and sink $t$. In that case the
polynomial $g$ is the sum of all monomials that are accepted by the
automaton (since it is acyclic, it accepts only finitely many). Let
$W$ denote the set of monomials accepted by the automaton $P$.  Then
the $(s,t)^{th}$ entry of the matrix
$f(M_{x_1},M_{x_2},\ldots,M_{x_n})$ is the polynomial
\[
\sum_{w\in W}f(w)w.
\]
\end{remark}

\begin{remark}
It is useful to combine the construction described in the previous
remark with \emph{substitution maps}. As above, let the ABP $P$ be a
deterministic finite substitution automaton with $q$ states accepting
monomials of degree at most $d$ over variables $X$ with start state
$s$ and accept state $t$. The substitutions are defined as follows:

For $1\le i,j\le q$, $\psi_{ij}:X\to Y^*$ is a substitution map
mapping variables in $X$ to monomials over $Y$, where $q$ is the
number of nodes in the ABP $P$. For each $x\in X$ define the matrix
$M'_x$ as follows:
\[
M'_x(i,j)= \psi_{ij}(x), 1\le i,j\le q.
\]

For every monomial $w=x_{j_1}x_{j_2}\ldots x_{j_d}$ accepted by $P$,
there is a unique $s$-to-$t$ path
$\gamma=(s,i_1),(i_1,i_2),\ldots,(i_{d-1},t)$ along which it
accepts. This defines the substitution map $\psi$:
\[
\psi(w)=\psi_{s,i_1}(x_{j_1})\psi_{i_1,i_2}(x_{j_2})\ldots
\psi_{i_{d-1},t}(x_{j_d})
\]
so that $\psi(w)\in Y^*$.

Let $W$ denote the set of monomials accepted by the automaton $P$.
Then the $(s,t)^{th}$ entry of the matrix
$f(M'_{x_1},M'_{x_2},\ldots,M'_{x_n})$ is the polynomial
\[
\sum_{w\in W}f(w)\psi(w).
\]

{From} the above considerations it is clear that if $f\in\F\angle{X}$
has a polynomial-size circuit and $P$ is a polynomial-size automaton
then $\sum_{w\in W}f(w)\psi(w)$ has a polynomial-size circuit.
\end{remark}

\subsection*{Comparing the reducibilities}

\begin{proposition}For noncommutative p-families $f=(f_n)$ and $g=(g_n)$ we have, 
 \begin{enumerate}
  \item $f\leq_{proj}g\Rightarrow f\leq_{iproj}g$ 
  \item $f\leq_{iproj}g\Rightarrow f\leq_{abp}g$ 
 \end{enumerate}

\end{proposition}

\begin{theorem}
 There are noncommutative p-families $f=(f_n)$ and $g=(g_n)$ such that
 $g\leq_{abp}f$ but $f \nleq_{iproj}g$ and $g \nleq_{iproj}f$.
\end{theorem} 

\begin{proof}
 We define the p-families as follows: $g_n,f_n \in
 \F\angle{x_1,x_2,\dots,x_n,y_1,\ldots,y_n}$ where $f_n=\prod_{i \in
   [n]}(x_i+y_i)$ and $g_n=x_1x_2\ldots x_n+y_1y_2\ldots y_n$. A key
 fact which is easy to check is that $g_n$ is irreducible for all $n$,
 and $f_n$ is a product of linear forms obviously. More crucially,
 $g_n$ has only two monomials for all $n$, whereas $f_n$ has $2^n$
 nonzero monomials.

Now, if $f \leq_{iproj} g$ then for some polynomial $p(n)$ and
substitution map $\phi$ we will have $g(\phi(X_{p(n)}))=f(X_n)$ where
$X_n=\{x_1,\ldots,x_n,y_1,\ldots,y_n\}$ and
$X_{p(n)}=\{x_1,\ldots,x_{p(n)},y_1,\ldots,y_{p(n)}\}$. However, the
substitution map cannot increase the number of monomials in
$g(\phi(X_{p(n)}))$ whereas $f(X_{n})$ has $2^n$ monomials. Hence $f
\nleq_{iproj}g$.

Also, $g \nleq_{iproj}f$ because for all $n$, $g_n$ is irreducible and
$f_n$ is a product of linear forms over $\F$.

Now, we claim $g \leq_{abp}f$, where the abp-reduction is defined by
the following matrix substitutions which are given by the following
DFA with start state $s$ and final state $t$:

\begin{itemize}
 \item In start state $s$, reading $x_1$ go to state $1$ and reading $y_1$ go to state $1'$.
 \item In state $i$, reading $x_{i+1}$ go to state $i+1$, $i<n-1$.
 \item In state $i'$, reading $y_{i+1}$ go to state $(i+1)'$, $i'<n-1$.
\item In state $n-1$, reading $x_{n}$ go to state $t$.
\item In state $(n-1)'$, reading $y_{n}$ go to state $t$.
\end{itemize}

For each variable in $\{x_1,\ldots,x_n,y_1,\ldots,y_n\}$ we substitute
matrices of dimension $2n \times 2n$, corresponding to the above DFA,
in the polynomial $f$ to obtain polynomial $g$.
\end{proof}

\begin{theorem}
 There are p-families $f$ and $g$ s.t $f \leq_{iproj} g$ but $f
 \nleq_{proj} g$.
\end{theorem}

\begin{proof}
 Let $f=\prod_{i \in [n]}(x_i+y_i)$ and $g=\prod_{i \in
   [n]}(z_0+z_1)$. Clearly, $f \leq_{iproj} g$ where the indexed
 projection will substitute $x_i$ for $z_0$ and $y_i$ for $z_1$ in the
 $i$-th linear factor $(z_0+z_1)$ of $g$. However, the usual
 $\leq_{proj}$ reduction cannot increase the number of variables in
 $g$ from two. Hence $f \nleq_{proj} g$.
\end{proof}

\section{Dyck Polynomials are $\vpnc$-complete}

\subsection{$\vpnc$-Completeness}

In this section we exhibit a \emph{natural} p-family which is
$\leq_{abp}$-complete for the complexity class $\vpnc$. We show that
any homogeneous degree $d$ polynomial $f\in \mathbb{F}\langle x_1,
x_2, \ldots, x_n \rangle$ computed by a non-commutative arithmetic
circuit of size $poly(n,d)$ is $abp$-reducible to the polynomials
$D_k$ for $k \geq 2$, where $D_k$ refers to the Dyck polynomial over
$k$ different types of brackets. Our main Theorem in this section can
be seen as an algebraic analogue of the Chomsky-Sch\"utzenberger
representation theorem \cite{cs63} (also see \cite[pg. 306]{Davis94}),
which says that every context-free language is a homomorphic image of
intersection of a language of balanced parenthesis strings over
suitable number of different types of parentheses and a regular
language.  More precisely,

\begin{theorem}[Chomsky-Sch\"utzenberger] 
A language $L$ over alphabet $\Sigma$ is context free iff there exist
\begin{enumerate}
\item a matched alphabet $P \cup \overline{P}$ ($P$ is set of $k$
  different types of opening parentheses $P= \{(_1, (_2, \ldots, (_k
  \}$ and $\overline{P}$ is the corresponding set of matched closing
  parentheses $\overline{P}= \{)_1, )_2, \ldots, )_k \}$),
\item a regular language $R$ over $P \cup \overline{P}$,
\item and a homomorphism $h :(P \cup \overline{P} )^* \mapsto \Sigma^*$ 
\end{enumerate}
such that $L= h(D \cap R)$, where $D$ is the set of all balanced parentheses strings over $P \cup \overline{P}$.
\end{theorem}

We now show that the p-family $\{D_{k,d}\}_{d\geq0}$ is
$\vpnc$-complete for $\leq_{abp}$ reductions, where the p-family
$\{D_{k,d}\}_{d\geq0}$, denoted $D_k$, is over set of $2k$ distinct
variables $\{(_i,)_i | 1 \leq i \leq k\}$ where $(_i$ and $)_i$ are
matching parenthesis pairs.  The polynomial $D_{k,d}$ consists of the
sum of all monomials $m$ which are well formed parenthesis strings of
degree $d$ over variables in $X_k$.

\begin{displaymath}
 D_{k,d} = \sum_{m \in W_{k,d}}m 
\end{displaymath}

where $W_{k,d}$ is set of well formed parenthesis strings of degree
$d$ over $X_k$.  The theorem we prove in this section is the
following.

\begin{theorem}
\label{thm:vpcomplete}
The Dyck polynomial $D_2=\{D_{2,d}\}_{d\geq0}$ is $\vpnc$-complete
under $\leq_{abp}$-reductions and hence $D_k=\{D_{k,d}\}_{d\geq0}$
 for $k\geq 2$ is $\vpnc$-complete under $\leq_{abp}$-reductions.
\end{theorem}

\begin{proof}
 Let $\{C_n\}_{n \geq 0}$ be a polynomial sized polynomial degree circuit family computing polynomials (by abuse of
   notation, also denoted by) $C_n$ in $\F \angle{x_1,\dots,x_n}$.
   Let $s(n)$ and $d(n)$ be polynomials bounding the size and degree
   of $C_n$, respectively. For each $n$ we will construct a collection
   of $2t(n)$ many matrices $M_1,M'_1,\ldots,M_{t(n)},M'_{t(n)}$ whose
   entries are either field elements or monomials in variables
   $\{x_1,\dots,x_n\}$ for a suitably large polynomial bound $t(n)$.
     These matrices have the property that polynomial $D_{t(n),q(n)}$,
     in which we substitute $M_i$ for $(_i$ and $M_i'$ for $)_i$,
     evaluates to a matrix
     $M=D_{t,q}(M_1,M'_1,\ldots,M_{t(n)},M'_{t(n)})$ whose top right
     corner entry is precisely the polynomial $C_n$.

     The idea underlying this construction is from the proof of the
     Chomsky-Sch\"utzenberger theorem (ours is an arithmetic version
     of it) : the matrices $M_1,M'_1,\ldots,M_t,M'_t$ actually
     correspond to the transitions of a deterministic finite state
     substitution automaton which will transform monomials of
     $D_{t(n),q(n)}$ into monomials of $C_n$ so that $M$'s top right
     entry (corresponding to the accept state) contains the polynomial
     $C_n$. We now give a structured description of the reduction.

\begin{enumerate}
 \item Firstly, we do not directly work with the circuit $C_n$ because
   we need to introduce a parsing structure to the monomials of
   $C_n$. We also need to make the circuit constant-free by
   introducing new variables (we will substitute back the constants
   for the new variables in the matrices). To this end, we will carry
   out the following modifications to the circuit $C_n$:

\begin{enumerate}
 \item For each product gate $f=gh$ in the circuit, we convert it to
   the product gate computing $f=(_fg)_fh$, where $(_f$ and $)_f$ are
   new variables.
\item We replace each input constant $a$ of the circuit $C_n$ by a
  degree-$3$ monomial $(_az_a)_a$, where $(_a, )_a, z_a$ are new
  variables.
\end{enumerate}

Let $C'_n$ denote the resulting arithmetic circuit after the above
transformations applied to the gates. The new circuit $C'_n$ computes
a polynomial over $\F \angle{X'}$ where

\begin{eqnarray*}
X' &= & X \cup \{(_g,)_g\mid g \text{ is a }\times\text{ gate in
  }C_n\}\\ & \cup & \{(_a,)_a\mid a \text{ is a constant in }C_n\}
    \\ & \cup & \{z_a \mid a \text{ is a constant appearing in }
      C_n\}.
\end{eqnarray*}

We make a further substitution: we replace every variable $y\in X$ by the degree-$2$ monomial $[_y]_y$ and every variable $z_a$ for constants $a$ appearing in $C_n$ by $[_{z_a}]_{z_a}$ to obtain the arithmetic circuit $C''_n$.

With these substitutions it is clear, by abuse of notation, that
$(C''_n)$ is a p-family. Furthermore, by construction $C''_n$ is a
polynomial whose monomials are certain properly balanced parenthesis
strings over the above parentheses set. It is not homogeneous, but
clearly its degree bounded by a polynomial in $(s(n)+d(n))$.
Furthermore, $C_n \leq_{abp} C''_n$ because we can recover $C_n$ by
substituting 1 for the parenthesis and $y$ for the term $[_y]_y$ and
the scalar $a$ for $[_{z_a}]_{z_a}$.

\item The next step is the crucial part of the proof. We describe the
  reduction from $C''_n$ to $D_{t(n)}$ for suitably chosen
  $t(n)$. Indeed, $t(n)$ is already the number of parentheses type
  used by $C''_n$, along with some additional parenthesis. Let the
  degree of polynomial $C''_n$ be $2r$.  Thus, monomials of $C''_n$
  are of even degree bounded by $2r$. We introduce $r+1$ \emph{new}
  parenthesis types $\{_j,\}_j$, $0 \leq j\leq r$ (to be used as
  prefix padding in order to get homogeneity) and consider the
  polynomial $D_{t(n),q(n)}$ where $q(n)=2r+2$ and $t(n)$ is $(r+1)$
  plus the number of parenthesis types occurring in $C''_n$.

The reduction will map all degree $2j$ monomials in $C''_n$ to
monomials in $D_{t,q}$ of the form
$m'=\{_1\}_1\{_2\}_2\ldots\{_{r-j}\}_{r-j}\{_0\}_0m$ where $m$ is a
  degree $2j$ monomial over the other parentheses types. Now $m'$ is
  of degree $2r+2$ for all choices of $j$ and it is clear that
  monomials which were distinct before the reduction remains distinct
  after the reduction.

Now the matrices of the automaton have to effect substitutions in
order to convert these $m'$ into a monomial of $C''_n$ of degree $2j$.
The strings accepted by this automaton is of the form $uv$, where
$u=\{_1\}_1\{_{2}\}_{2}\ldots\{_{i-1}\}_{i-1}\{_0\}_0$, $0\leq i\leq r+1$ and $v$
is a well-balanced string over remaining parentheses type. This
automaton is essentially the one defined in the proof of the
Chomsky-Sch\"utzenberger theorem.  We outline its description. The
automaton runs only on monomials of $D_{t,q}$ and hence can be seen as
a layered DAG with exactly $q(n)$ layers.

\begin{enumerate}
 \item 

The start state of the automaton is $(\hat{s},0)$.
The automaton first looks for prefix
$\{_1\}_1\{_2\}_2\ldots\{_{r-j}\}_{r-j}\{_0\}_0$. As it reads these
variables, one by one, it steps through states $(\hat{s},i)$,
substitutes 1 for each of them, and reaches state $(s,2(r-j+1))$ when
it reads $\}_0$, where $s$ is the name of the output gate of circuit
$C''_n$. If any of $\{_l,\}_l$, $l \in [r] \cup \{0\}$ occur later
they are substituted by 0 (to kill that monomial).

\item The automaton will substitute $[_x]_x$ by $x$ (if $[_x$ is not
  immediately followed by $]_x$ then it substitutes $0$ for $[_x$).
  Similarly, the automaton substitutes $[_a]_a$ by $a$ (if $[_a$ is
    not followed by $]_a$ then it substitutes $0$ for it).

\item Now, we describe the crucial transitions of the automaton
  continuing from state $(s,2(r-j+1))$, where $s$ is the output gate
  of circuit $C''_n$. The transitions are defined using the structure
  of the circuit $C''_n$. At this point the automaton is looking for a
  degree $2j$ monomial. Let $D <2r+2$.  We have the following
  transitions:

\begin{enumerate}
\item $(\hat{s},2j) \rightarrow \{_{j+1}\}_{j+1}(\hat{s},2(j+1))$, where $0\leq j
  < r$
\item $(\hat{s},2(r-j)) \rightarrow \{_0\}_0(s,2(r-j+1))$, where $0\leq j
  \leq r$ and $s$ is the output gate in the circuit $C''_n$.
 \item $(g,D)\rightarrow (_g(g_l,D+1)$, where $g$ is an internal
   product gate in circuit $C''_n$ and $g_l$ is its left child.

\item Include the transition $(g,D)\rightarrow (_h(h_l,D+1)$, if $g$
  is an internal $+$ gate in circuit $C''_n$, $h$ is an internal
  product gate such that there is a directed path of $+$ gates from
  $h$ to $g$. As before, $h_l$ denotes the left child of $h$.

\item For each input variable, say $z$, in the circuit $C''_n$ and for
  each product gate $g$ in the circuit $C''_n$, the automaton includes
  the transition $(h,D) \rightarrow [_z]_z)_g(g_r,D+3)$, if
  $D+3<2r+2$, where $g_r$ is the right child of the internal product
  gate $g$, and $h$ stands for any internal gate in $C''_n$.

If $D+3=2r+2$ then the automaton instead includes the transition
  $(h,D) \rightarrow [_z]_z)_g(t,2r+2)$, where $(t,2r+2)$ is the
  unique accepting state of the automaton.
\end{enumerate}

Note that the interpretation of the transition
\[
(h,D) \rightarrow [_z]_z)_g(g_r,D+3)
\]
is as follows: The automaton reads the degree-$3$ monomial $[_z]_z)_g$
and goes from state $(h,D)$ to $(g_r,D+3)$.
\end{enumerate}
\end{enumerate}

We now describe the matrices that we substitute for each
parenthesis. Let $M_p$ be the matrix we substitute for parenthesis $p$
its whose rows and columns are labelled by nodes of the ABP.

We define the matrix $M_p$ for parenthesis $p$ as follows:

\begin{displaymath}
m_{u,v}=M_p[u,v]= \left\{ \begin{array}{ll} 1 & \textrm{if $p \in U$
    and $\exists e=(u,v) \in$ E(A) and label of $e$ is $p$ }\\ z &
  \textrm{if $p =]_z$ and $\exists e=(u,v) \in$ E(A) and label of $e$
  is $p$ }
\end{array} \right.
\end{displaymath}
where $z$ denotes a variable in the circuit $C''_n$ and 
 E(A) is the edge set of the automaton A and 
\begin{eqnarray*}
U &=&\{[_z \mid z \text { is a variable in $C''_n$} \} \\ & \bigcup & \{(_i,)_i \mid i \in [s']\}\\ & \bigcup & \{\{_j,\}_j \mid j\in [r] \cup \{0\}\} 
\end{eqnarray*}
where $s'$ denotes the number of product gates in the circuit $C_n$.


It is clear that after substituting these matrices for the variables
in the polynomial $D_k$, where $k$ denotes the number of parenthesis
types in $C''_n$, the top right corner entry of the resulting matrix
is polynomial computed by the given circuit $C$. It is easy to see
that $D_2 \leq_{abp} D_k$ for all $k > 2$. Furthermore, we can show
for any $k>2$ that $D_k \leq_{abp}D_2$, by suitably encoding different
types of brackets into two types. Thus, it follows that the p-family
$D_k$, for any $k\geq 2$, is $\vpnc$-complete under
$\leq_{abp}$-reductions.
\end{proof}

\begin{remark}
We note that $D_1 \abp \pal \abp D_2$ and $D_2 \not\abp \pal \not\abp
D_1$. To see this the first one, observe that we have a DFA (of
growing size) for $D_1$. Hence $D_1$ is in $\vbpnc$ which trivially
implies that $D_1$ is $\abp$-reducible to $\pal$. As $\pal$ is not in
$\vbpnc$ \cite{N91}, it follows that $\pal \not\abp D_1$. We show in
theorem \ref{pal-to-d2} that $D_2$ is not $\abp$-reducible to $\pal$.
\end{remark}

\section{Palindrome Polynomials are $\vsk$-complete}

\begin{theorem}\label{pal-vsk}
The p-family $\PAL$ is $\vsk$-complete for $\leq_{abp}$ reductions.
\end{theorem}

\begin{proof}
The proof is along the same lines as that of Theorem
\ref{thm:vpcomplete}. We will show for any p-family in $\vsk$ is
$\leq_{abp}$-reducible to $\PAL$.

Let $\{C_n\}_{n \geq 0}$ be a polynomial sized skew circuit family of
  polynomial degree $d(n)$ computing polynomials (by abuse of
  notation, also denoted by) $C_n$ in $\F \angle{x_1,\dots,x_n}$.  Let
  $s(n)$ and $d(n)$ be polynomials bounding the size and degree of
  $C_n$, respectively. We will construct a collection of $2t(n)$
  matrices $M_1,M'_1,\ldots,M_{t(n)},M'_{t(n)}$ whose entries are
  either field elements or monomials in variables
  $\{x_{1,L},x_{1,R},\dots,x_{n,L},x_{n,R}\}$ for a suitably large
  polynomial bound $t(n)$. These matrices have the property that
  polynomial $\PAL_{t(n)}$, in which we substitute $M_i$ for $x_{i,L}$
  and $M_i'$ for $x_{i,R}$, evaluates to a matrix
  $M=\PAL_t(M_1,M'_1,\ldots,M_{t(n)},M'_{t(n)})$ whose top right
  corner entry is precisely the polynomial $C_n$.

As in the proof of Theorem \ref{thm:vpcomplete}, the basic idea is
from the Chomsky-Sch\"utzenberger theorem: the matrices
$M_1,M'_1,\ldots,M_t,M'_t$ will correspond to the transitions of a
deterministic finite state (substitution) automaton which will
transform monomials of $\PAL_{t(n)}$ into monomials of $C_n$ so that
$M$'s top right entry (corresponding to the accept state) contains the
polynomial $C_n$. We now give a structured description of the
reduction.

W.l.o.g we can assume the skew circuit $C_n$ is homogeneous. At the
input level, we replace variables $x$ by $x_Lx_R$. 

\begin{enumerate}
\item Firstly, we do not directly work with the circuit $C_n$ because
  we need to introduce a parsing structure to the monomials of
  $C_n$. We also need to make the circuit constant-free by introducing
  new variables (we will substitute back the constants for the new
  variables in the matrices). To this end, we will carry out the
  following transformations:
 
\begin{enumerate}
 \item For each left-skew product gate $f=xh$ in the circuit $C_n$
   (similarly for the right-skew gate $f=hx$), where $x$ is an input
   variable and $h$ a gate in the circuit, let $e=(h,f)$ denote the
   directed edge in the circuit $C_n$ (seen as a directed acyclic
   graph). We convert it to the gates

\begin{eqnarray*}
f'&=&hx_{(e,h,R)} \\
f^{''}&=& x_{(e,h,L)}f',
\end{eqnarray*}

where $x_{(e,h,L)},x_{(e,h,R)}$ are fresh variables.

\item For each product gate $f=ah$ in the circuit $C_n$ for $a \in \F$
  and $e=(h,f)$ is the edge in the circuit we convert it to gates

\begin{eqnarray*}
f'&=&ha_{(e,h,R)} \\
f^{''}&=&a_{(e,h,L)}f'
\end{eqnarray*}

where $a_{(e,h,L)},a_{(e,h,R)}$ are fresh variables.
\end{enumerate}

Let $C'_n$ denote the resulting circuit. It computes a polynomial over
$\F \angle{X'}$ where the variable set $X'$ is:

\begin{eqnarray*}
 X'& = &\{x_{(e,h,R)},x_{(e,h,R)} | x \in X, e \in E\} \\ & \cup &
   \{a_{(e,h,L)},a_{(e,h,R)} | a \text{ is a constant appearing in the
       edge $e \in E$ }\}.
\end{eqnarray*}

Here $E$ is set of all edges $e$ in the given circuit $C_n$.

Clearly, $(C'_n)$ is a p-family, and $C'_n$ is a polynomial whose
nonzero monomials $m$ are palindrome monomials in the following sense:
\textrm{ in a monomial $m$ of degree $2d$, for all $i \in [d]$ and for
  any edge $e$ and gate $g$ at position $i$ we have variable
  $x_{(e,g,L)}$ and at position $2d-i+1$ we have variable
  $x_{(e,g,R)}$}.

We also have the reduction $(C_n) \leq_{abp} (C'_n)$ because we can
recover $C_n$ from $C'_n$ by substituting $x$ for either $x_{e,h,L}$
or $x_{e,h,R}$ (and 1 for the other variable) and the scalar $a$ for
either $a_{e,h,L}$ or $a_{e,h,R}$ (and 1 for the other
variable). Notice that the number of variables in $C'_n$ and the degree of
$C'_n$ are polynomially bounded by a suitable function of $n$ (but we
are not specifying it for ease of notation).

\item Let the degree of polynomial $C'_n$ be $2r$. Thus monomials of
  $C'_n$ are of even degree bounded by $2r$. Like in Theorem
  \ref{thm:vpcomplete}, we will introduce $r+1$ \emph{new} variable
  pairs $y_{j,L},y_{j,R}$, $0 \leq j\leq r$ (to be used as prefix and
  suffix padding in order to get homogeneity).  The reduction will map
  a degree $2j$ monomial $m$ in $C'_n$ to monomial $m'$ in
  $\PAL_{r+1}$ of the following form:
\begin{align*}
 m'&=(y_{1,L}y_{2,L}\ldots y_{r-j,L}y_{0,L})m(y_{0,R}y_{r-j,R}\ldots
 y_{2,R}y_{1,R})
\end{align*}
Now, $m'$ is of degree $2r+2$ for all choices of $j$ and it is clear
that monomials which were distinct before the reduction remains
distinct after the reduction. Let $C''_n$ denote this resulting new
circuit.

\item Like in Theorem \ref{thm:vpcomplete}, we construct automaton A
  from this modified circuit $C''_n$. We construct automaton which
  (apart from accepting many non-palindrome monomials) accepts only
  palindrome monomials $ww^{R}$ such that the first half $w$ is
  ``compatible'' with the circuit structure of $C''_n$ (and
  monomials whose first half is non-compatible are not accepted by the
  automaton A). Now the matrices of the automaton have only to effect
  substitutions in a careful manner to convert these $m'$ into a
  monomial of $C''_n$ of degree $2j$. The automaton is a layered DAG
  with exactly $2r+2$ layers.

\begin{enumerate}
\item The start state of the automaton is $(\hat{s},0)$. The automaton
  first looks for a prefix $(y_{1,L}y_{2,L}\ldots
  y_{r-j,L}y_{0,L})$. As it reads these variables, one by one, it
  steps through states $(\hat{s},i)$, substitutes 1 for each of them,
  and reaches state $(s,(r-j+1))$ when it reads $y_{0,L}$, where $s$
  is the name of the output gate of circuit $C''_n$. If any of
  $y_{l,L}$, $l \in [r] \cup \{0\}$ occur later they are substituted
  by 0 (to kill that monomial).

  \item Now we describe the transitions of the automaton continuing
    from state $(s,(r-j+1))$. Here the automaton has to use the
    structure of the circuit $C''_n$ to define further
    transitions. At this point the automaton is looking for a degree
    $2j$ monomial. Let $D <2r+2$.  We have the following transitions:

\begin{enumerate}

\item $(\hat{s},j) \rightarrow y_{(j+1,L)}(\hat{s},j+1))$, where
  $0\leq j < r$ (as already described above).

\item $(\hat{s},j) \rightarrow y_{(0,L)}(s,j+1)$, where $0\leq j \leq
  r$ and $s$ is the output gate in the circuit $C''_n$.

\item In state $(s,j+1)$ if the automaton reads variable $x_{e,g,L}$
  (or variable $a_{e,g,L}$) then it moves to state $(g,j+2)$ if the
  gate $g$ is a left-skew multiplication occurring in the circuit
  $C''_n$, and the directed path from $g$ to $s$ in the circuit has
  only $+$ gates or right-skew multiplication gates in it. Formally,
  the transition made is:

\begin{align*}
(s,j+1)&\rightarrow x_{(e,g,L)}(g,j+2).
\end{align*}
We have a similar transition when the automaton reads variable
$a_{e,g,L}$.

\item In general, when the automaton is in state $(g,D)$ for a
  left-skew multiplication gate $g$ in the circuit and it reads
  variable $x_{e,h,L}$ (or $a_{e,h,L}$) then it moves to state
  $(h,D+1)$ if the gate $h$ is left-skew occurring in the circuit, and
  the directed path from $h$ to $g$ has only $+$ gates or right-skew
  multiplication gates in it. Formally, the transition made is:

\begin{align*}
(g,D)&\rightarrow x_{(e,h,L)}(h,D+1).
\end{align*}

We have a similar transition for variable $a_{e,h,L}$.

\item Proceeding thus, when the automaton reaches a state $(g,r+1)$
for some left-skew multiplication gate it makes only transitions
of the form: 

\begin{align*}
(g,D)&\rightarrow x_{(e,h,R)}(t,D+1),
\end{align*}

for all variables $x_{e,h,R}$ and for all $D<2r+2$. The state
$(t,2r+2)$ is the unique accepting state of the automaton.
\end{enumerate}

Transitions (i-iv) reads the first half of any input monomial which
are compatible with the structure of the circuit $C''_n$.   By
construction of the transitions in (i-iv) the following claim holds.

\begin{claim}
 The \textrm{DFA} defined above accepts a palindrome string $uv \in
 (X')^{2r+2}$ iff the palindrome $uv$ is a nonzero monomial in the
 polynomial computed by $C''_n$.
\end{claim}
\end{enumerate}

\item We can convert this automaton into a homogeneous ABP $A$
  computing the homogeneous polynomial of degree $2r+2$. We now
  describe matrices we substitute for each variable. Let $M_z$ be the
  matrix we substitute for a variable $z$ where rows and columns of
  $M_z$ are labelled by nodes of the ABP.

We set entries of the matrix $M_z$ for a variable $z$ as follows:
\begin{itemize}
 \item If the variable $z =a_{(e,h,L)}$ where $a$ is a scalar
   appearing on the edge $e$ in the circuit $C_n$, then we set
   $m_{u,v}=M_z[u,v]=a$ iff the automaton reaches the state $v$ from the state $u$ when it reads $z$.
 \item Else, if the variable $z =a_{(e,h,R)}$ where $a$ is a scalar
   appearing on the edge $e$ in the circuit $C_n$, then we set
   $m_{u,v}=M_z[u,v]=1$ iff the automaton reaches the state $v$ from the state $u$ when it reads $z$.

\item Else, if $z = x_{(e,g,L)}$, where $x \in X$, $e$ is an edge in the circuit $C_n$, $\text{ g
  is some gate in $C_n$}$, then
\begin{itemize}

 \item If the actual variable for $z$ occurs as left multiplication on
   the edge $e$, then we set $m_{u,v}=x$ iff the automaton reaches the state $v$ from the state $u$ when it reads $z$.

\item Else, if $m_{u,v}=1$ (i.e., the actual variable for $z$ occurs as right multiplication )
\end{itemize}
 \item Else, if $z = x_{(e,g,R)}$, where $x \in X$, $e$ is an edge in the circuit $C_n$, $\text{ g
   is some gate in $C_n$}$, then

\begin{itemize}

 \item If the actual variable for $z$ occurs as right multiplication
   on the edge $e$, then we set $m_{u,v}=x$ iff the automaton reaches the state $v$ from the state $u$ when it reads $z$.
\item Else, if $m_{u,v}=1$ (i.e., the actual variable for $z$ occurs as left multiplication )
\end{itemize}

 \item Else, if the variable $z = y_{(j,L)}$ or $z = y_{(j,R)}$, $0\leq j \leq r$ then we set
   $m_{u,v}=M_z[u,v]=1$ iff the automaton reaches the state $v$ from the state $u$ when it reads $z$.
\item Else, we set $m_{u,v}=0$.

\end{itemize}
\end{enumerate}

It is clear that on substituting these matrices for the variables in $\PAL_{r+1}$ , we get the polynomial computed by the given circuit $C_n$ in the
top right corner entry of the resulting matrix. This completes the proof.
\end{proof}

\section{A Ladner's Theorem analogue for $\vnpnc$}

In this section we explore the structure of $\vnpnc$ assuming the
sum-of-squares conjecture. The sum-of-squares conjecture implies that
the p-family $ID$ (which is in $\vnpnc$) is not in $\vpnc$ \cite{HWY10a}. In
particular, the conjecture implies that $\vpnc\neq \vnpnc$. A natural
question that arises is whether this conjecture implies that there are
p-families in $\vnpnc\setminus \vpnc$ that are not $\vnpnc$-complete.

This is similar in spirit to the well-known Ladner's Theorem that
shows, assuming $\P\ne \NP$, that there is an infinite hierarchy of
polynomial degrees between $\P$ and $\NP$-complete. For commutative
Valiant's classes, the existence of $\VNP$-intermediate p-families is
investigated by B\"urgisser \cite{Burg99}. The results there require
an additional assumption about counting classes in the boolean
setting.

\begin{conjecture}[$SOS_k$ Conjecture]
Consider the question of expressing the biquadratic polynomial
\[
SOS_k(x_1,\ldots,x_k,y_1,\ldots,x_k)=(\sum_{i\in
  [k]}x_i^2)(\sum_{i\in [k]}y_i^2)
\]
 as a sum of squares $(\sum_{i\in [s]}f_i^2)$, where $f_i$ are all
 homogeneous bilinear polynomials with the minimum $s$.

The $SOS_k$ conjecture states that over the field of complex numbers
$\C$, for all $k$ we have the lower bound $s=\Omega(k^{1+\epsilon})$.
\end{conjecture}

In \cite{HWY10a}, it is shown that the $SOS_k$-conjecture implies that
the p-family $ID=\{ID_d\}_{d\geq 0}$ where $ID_d(x_0,x_1)=\sum_{w \in
  \{x_0,x_1\}^d}ww$ is not in $\vpnc$. In fact, they prove exponential
circuit size lower bounds for $ID_d$ assuming the conjecture. We need the
following definition.

\begin{definition}[$\vnpnc$-intermediate]
 We say that a noncommutative p-family $f=(f_n)_{n\geq 0}$ is
 $\vnpnc$-intermediate if $f \notin \vpnc$ and $f$ is not
 $\vnpnc$-complete w.r.t. $\leq_{iproj}$ reductions.
\end{definition}

In this section, we show the $SOS_k$ conjecture actually yields much
more inside $\vnpnc$. We prove the following results.

\begin{enumerate}
 \item That $ID$ is a $\vnpnc$-intermediate polynomial assuming
   $SOS_k$ conjecture.
\item There are infinitely many p-families $f^{(i)}$, $i=1,2,\ldots$
  in $\vnpnc$ such that for all $i$, $f^{(i)} \leq_{iproj} f^{(i+1)}$ and
  $f^{(i+1)} \nleq_{iproj} f^{(i)}$.
\end{enumerate}

We do not have similar results for the stronger $\leq_{abp}$
reducibility.

The proof of the first result is by using a simple "transfer" theorem
which allows us to transfer a $\vnpnc$-complete p-family w.r.t
$\leq_{iproj}$ reductions to a commutative $\VNP$-complete p-family
w.r.t $\leq_{proj}$ reductions. 

\begin{definition}
Let $f=(f_n)$ be a p-family in $\vnpnc$, where each $f_n$ is a
homogeneous polynomial of degree $d(n)$. We define the
\emph{commutative version} $f^{(c)}=(f^{(c)}_n)$ as follows:
Suppose $f_n \in \F\angle{X_n}$. Let $Y_n=\bigcup_{1\leq i \leq
  d(n)}X_{n,i}$ be a new variable set where $X_{n,i}=\{x_{ji}|\forall
x_j \in X_n\}$ is a copy of the variable set $X_n$ for the $i^{th}$
position. If the polynomial $f_n=\sum \alpha_mm$ where $\alpha_m \in
\F$ and $m \in X_n^{d(n)}$ is a monomial, the polynomial $f^{(c)}_n$
is defined as $f^{(c)}_n=\sum \alpha_mm'$, where if
$m=x_{j_1}x_{j_2}\ldots x_{j_d}$ then $m'=x_{j_1,1}x_{j_2,2}\ldots
x_{j_d,d}$.
\end{definition}

Clearly, $f^{(c)}_n \in \F[X]$ and is a set-multilinear homogeneous polynomial of
 degree $d(n)$.

\begin{lemma}\label{lem:transfer}
For any p-families $f$ and $g$, if $f \leq_{iproj} g$ then $f^{(c)}
\leq_{proj}g^{(c)}$.
\end{lemma}

\begin{proof}
 Since $f\leq_{iproj}g$, for every $n$ there is a polynomial $p(n)$
 and an indexed projection $\phi_n:[d_{p(n)}] \times
 X_{p(n)}\rightarrow (Y_{ij})_{1\leq i,j\leq n}$ s.t.
 $f_n(Y_n)=g(\phi_n(X_{p(n)}))$ where $d_{p(n)}$ is the degree of the
 polynomial $g_{p(n)}$. Define $\phi'_n:\bigcup_{i \in [d(n)]}
 X_{p(n),i}\rightarrow Y_n$ as $\phi'_n(x_{ji})=\phi_n(i,x_j)$ for
 $1\leq i,j \leq n$. Clearly, $f^{(c)}$ is reducible to $g^{(c)}$ via
 this projection reduction. This completes the proof.
\end{proof}

The following theorem is a corollary of Lemma \ref{lem:transfer}.

\begin{theorem}[Transfer theorem]
\label{thm:transfer}
 Let $f=(f_n)\in \vnpnc$ be a homogeneous p-family that is
 $\vnpnc$-complete for $\leq_{iproj}$-reductions. Then $f^{(c)} \in
 \VNP$ is $\VNP$-complete for $\leq_{proj}$-reductions.
\end{theorem}

\begin{proof}
 Since $\PER \leq_{iproj}f$, by Lemma \ref{lem:transfer}
 $\PER^{(c)}_{d}\leq_{proj}f^{(c)}$. This completes the proof of the
 theorem.
\end{proof}

\begin{theorem}\label{thm:notcomplete}
 The polynomial $ID$ is not $\vnpnc$-complete under
 $\leq_{iproj}$-reductions.
\end{theorem}

\begin{proof}
 Suppose, to the contrary that $ID$ is $\vnpnc$-complete w.r.t
 $\leq_{iproj}$-reductions. Then $\PER \leq_{iproj} ID$. Define the
 noncommutative p-family $ID'=(ID'_n)_{n\geq0}$, where $ID'_n \in
 \F\angle{X_n}$ where
 $X_n=\{x_{0,1},x_{0,2},\ldots,x_{0,n},x_{1,1},x_{1,2},\ldots,x_{1,n}\}$
 and 

\[
ID'_n=\sum_{z_i \in \{x_{0,i},x_{1,i}\}, i \in [n]}z_1z_2\ldots z_nz_1\ldots z_n.
\]

Clearly, $ID \leq_{iproj}ID'$. Hence $\PER \leq_{iproj}ID'$. Applying
the transfer theorem (Theorem \ref{thm:transfer}), we have that $\PER
\leq_{proj}ID'^{(c)}$ in the commutative setting. However,
$ID'^{(c)}=\prod_{i \in
  [n]}(x_{0,i}x_{0,n+i}+x_{1,i}x_{1,n+i})$. Thus, $ID'^{(c)}$ is a
reducible polynomial with factors of degree 2. Since $\PER_n$ is
irreducible for all $n$, it follows that $\PER$ cannot be
$\leq_{proj}$ reducible to $ID'$.
\end{proof}

Assuming the $SOS_k$ conjecture, Theorem \ref{thm:notcomplete} implies
that $ID$ is a $\vnpnc$-intermediate polynomial.

\begin{corollary}
 Assuming $SOS_k$ conjecture, $ID \notin \vpnc$ and $ID$ is not
 $\vnpnc$-complete under $\leq_{iproj}$-reductions.
\end{corollary}

Now we will show that there
are infinitely many p-families $f^{(i)}$ such that
$f^{(i)}\leq_{iproj}f^{(i+1)}$ but for all $i$
$f^{(i+1)}\nleq_{iproj}f^{(i)}$. For that we need the following
observation that $ID$ is not even $\vpnc$-hard
w.r.t.\ $\leq_{iproj}$-reductions.

\begin{theorem}\label{id-not-hard}
 The p-family $ID$ is not $\vpnc$-hard w.r.t $\leq_{iproj}$-reductions.
\end{theorem}

\begin{proof}
 We will prove that the Dyck p-family $D_2$ is not
 $\leq_{iproj}$-reducible to $ID$. Suppose $D_2\leq_{iproj}ID$. Since
 the reduction is an indexed projection it follows that the polynomial
 family $\hat{D}_2$ defined below is also $\leq_{iproj}$-reducible
 to $ID$ by essentially the same
 reduction. $\hat{D}_2=(\hat{D}_{2,n})$, where
 $\hat{D}_{2,n}$ is a homogeneous degree $2n$ polynomial on
 variable set of size $4n$ $\{(_i,)_i,[_i,]_i| i \in [n]\}$ where
 $(_i,)_i,[_i$ and $]_i$ are variables that can occur only in $i$-th
 position. The polynomial $\hat{D}_{2,n}$ is defined as an indexed
 projection of $D_{2,n}$ obtained by replacing the $i$-th occurrence of
 a bracket $b\in\{(,),[,]\}$ by its indexed version
 $b_i\in\{(_i,)_i,[_i,]_i\}$. We observe that the p-families
 $\hat{D}_2$ and $D_2$ are $\leq_{iproj}$-reducible to each other.

Now, by assumption $\hat{D}_2\leq_{iproj}ID\leq_{iproj}ID'$ which
means that, by the transfer theorem (Theorem \ref{thm:notcomplete}),
that the commutative version
$\hat{D}_2^{(c)}\leq_{proj}ID'^{(c)}$. Now, we know for all $n$
that $ID_n'^{(c)}=\prod_{i \in
  [n]}(x_{0,i}x_{0,n+i}+x_{1,i}x_{1,n+i})$. We show in the following
claim that the commutative polynomials $\hat{D}^{(c)}_{2,n}$ are
irreducible which rules out $\hat{D}^{(c)}_2\leq_{proj}ID'^{(c)}$,
and hence completes the proof by contradiction.

\begin{claim}
 The polynomial $\hat{D}^{(c)}_{2,n}$ is irreducible for each $n$.
\end{claim}

\noindent\textit{Proof of Claim:}~~Suppose $\hat{D}^{(c)}_{2,n}=g.h$
is a nontrivial factorization. We will derive a contradiction. First,
note that $\hat{D}^{(c)}_{2,n}$ is set-multilinear of degree $2n$
where the $i$-th location is allowed only variables from the set
$\{(_i,)_i,[_i,]_i\}$. Since $\hat{D}^{(c)}_{2,n}$ is multilinear, it
follows that both $g$ and $h$ are homogeneous multilinear and
$Var(g)\cap Var(h)=\emptyset$, where $Var(g),Var(h)$ are the variables
sets of $g$ and $h$ respectively.

Thus, every nonzero monomial $m$ of $f$ has a unique factorization
$m=m_1m_2$, where $m_1$ occurs in $g$ and $m_2$ in $h$. There are no
cancellations of terms in the product $gh$. Hence, it also follows
that both $g$ and $h$ are set-multilinear, where the set of locations
$[2n]$ is partitioned as $S$ and $[2n]\setminus S$ and the monomials
of $g$ are over variables in $\{(_i,)_i,[_i,]_i | i \in S\}$ and $h$'s
monomials are over variables in $\{(_i,)_i,[_i,]_i | i \in
[2n]\setminus S\}$. Now, there are monomials $m$ occurring in
$\hat{D}^{(c)}_{2,n}$ such that the projection of $m$ onto
positions in $S$ does not give a string of matched brackets. Let $m$
be any such monomial. Then we have the factorization $m=m_1.m_2$,
where $m_1$ and $m_2$ are monomials that occur in $g$ and $h$
respectively. Let the monomial $m'$ be obtained from $m$ by swapping
$(_i$ with $[_i$ and $)_i$ with $]_i$. Then $m'=m_1'm_2'$, where
$m_1'$ and $m_2'$ occur in $g$ and $h$, respectively.

Now, since there are no cancellations in the product $gh$, the
monomial $m_1'm_2$ (which is not a properly matched bracket string)
must also occur in $gh$ and hence in $\hat{D}^{(c)}_{2,n}$, which is a
contradiction. This completes the proof of the claim and hence the
theorem.
\end{proof}

We have shown that $ID$ is $\vnpnc$-intermediate assuming $SOS_k$
conjecture. On the other hand, $D_2 \nleq_{iproj}ID$
unconditionally. Our aim is to use $D_2$ and $ID$ to create an
infinite collection $f^{(i)}$ of p-families in $\vnpnc$ such that
$f^{(i)} \leq_{iproj}f^{(i+1)}$ but $f^{(i+1)}
\nleq_{iproj}f^{(i)}$. 

Let $ID=(ID_n)$ where $ID_n$ are degree $2n$, and
$D_2=(D_{2,n})_{n\geq 0}$ where $D_{2,n}$ are degree $2n$.

\begin{itemize}
 \item Define $f^{(1)}=ID$.
\item $f^{(2)}=(f_n^{(2)})$ where $f_n^{(2)}=D_{2,n}ID_{n}$.

\item $f^{(i)}=(f_n^{(i)})=(D_{2,n}ID_{n}\ldots D_{2,n}ID_{n})$, where
  $f_n^{(i)}=f_n^{(i-1)}D_{2,n}ID_{n}$ for all $i$ and $n$.
\end{itemize}

Clearly, $f^{(i)} \in \vnpnc$ for all $i$. 

\begin{proposition}
 For every $i$, $f^{(i)} \leq_{iproj} f^{(i+1)}$, where the $f^{(i)}$
 are the p-families defined above.
\end{proposition}

\begin{proof}
We explain the easy proof for $f^{(1)}\leq_{iproj} f^{(2)}$ which can
be easily extended to all $i$. The indexed projection that gives a
reduction from $f_n^{(1)}$ to $f_n^{(2)}$ will simply substitute $1$
for the variables $($ occurring in positions $1\le i\le n$, and $1$
for the variables $)$ occurring in positions $n+1\le i\le 2n$. For all
other occurrences of the variables of $D_{2,n}$ the indexed projection
substitutes $0$. This substitution picks out the following unique
degree-$2n$ monomial in $D_{2,n}$
\[
\underbrace{(((\cdots((}_{n-times}\underbrace{))\cdots)))}_{n-times}
\]
 in the polynomial $D_{2,n}$ and gives it the value $1$, and it zeros
 out the remaining monomials of $D_{2,n}$.

Finally, the indexed projection substitutes $x$ for $x$, for each
variable $x$ occurring in the polynomial $ID_n$.
\end{proof}

\begin{theorem}
Assuming the $SOS_k$-conjecture, for every $i$, we have $f^{(i+1)}
\nleq_{iproj} f^{(i)}$.
\end{theorem}

\begin{proof}
 Suppose to the contrary that $f^{(i+1)} \leq_{iproj} f^{(i)}$. Then
 there is a polynomial $p(n)$ and indexed projection map $\phi_n$ s.t
 $f_{p(n)}^{(i)}(\phi_n(X_{p(n)}^{(i)}))=f_{n}^{(i+1)}(X_{n}^{(i+1)})$,
 where $X_{p(n)}^{(i)}=Var(f_{p(n)}^{(i)})$ and
 $X_{n}^{(i+1)}=Var(f_{n}^{(i+1)})$. Now, we will derive a
 contradiction from this. We have:
\begin{itemize}
 \item $f_{p(n)}^{(i)}=\underbrace{D_{2,p(n)}ID_{p(n)}\ldots D_{2,p(n)}ID_{p(n)}}_{i-times}$
\item $f_{n}^{(i+1)}=\underbrace{D_{2,n}ID_{n}\ldots D_{2,n}ID_{n}}_{(i+1)-times}$
\end{itemize}
Since $ID_n \nleq_{iproj} D_{2,n}$ (by \cite{HWY10a} assuming
$SOS_k$-conjecture), we have $D_{2,n}ID_n \nleq_{iproj} D_{2,p(n)}$
and $D_{2,n}ID_n \nleq_{iproj} ID_{2,p(n)}$ because of irreducibility
of $\hat{D}^{(c)}_{2,n}$ (as shown in Theorem
\ref{thm:notcomplete}). Hence $D_{2,n}ID_n$ must get mapped by the
projection $\phi_n$ to the product $D_{2,p(n)}ID_{p(n)}$ or
$ID_{p(n)}D_{2,p(n)}$, overlapping both factors. But $f^{(i+1)}_n$ has
$(i+1)$ such factors $D_{2,n}ID_n$. Hence, at least one of these
factors $D_{2,n}ID_n$ must map wholly to $ID_{p(n)}$ or $D_{2,p(n)}$
by the indexed projection $\phi_n$. If $D_{2,n}ID_n$ maps to
$ID_{p(n)}$ that contradicts Theorem~\ref{id-not-hard}. If
$D_{2,n}ID_n$ maps to $D_{2,p(n)}$ then $ID_n$ must be in $\vpnc$,
which is not true assuming the $SOS_k$ conjecture.
\end{proof}

\section{Inside $\vpnc$}

We first show that $D_2$ is strictly harder than $\PAL$ w.r.t
$\leq_{abp}$-reductions.

\begin{theorem}\label{pal-to-d2}
 $\PAL \leq_{abp}D_2$ but $ D_2\nleq_{abp} \PAL$.
\end{theorem}

\begin{proof}
 As $\PAL$ has polynomial size circuit, clearly $\PAL \leq_{abp}D_2$ since $D_2$ is $\vpnc$-complete. For clarity, we give a direct reduction below. Consider $\PAL_n=\sum_{w \in
   \{x_0,x_1\}^n}w.w^R$ and $D_{2,n}$. The idea is to encode monomial
 $ww^R$ by encoding $x_0$ as $($ and $x_1$ as $[$ for position $i \in
   [n]$ and $x_0$ as $)$ and $x_1$ as $]$ for position $i \in
 [n+1,2n]$. We can easily design an automaton with $O(n)$ states that
 replaces $($ in $i$-th position by $x_0$ and $[$ in $i$-th position
   by $x_1$ for $i \in [n]$ and if it sees a closing bracket in any
   positions $i \in [n]$ it replaces it by $0$. Similarly, the position
   from $n+1,\ldots,2n$ are handled by replacing $)$ in $i$-th
   position by $x_0$ and $]$ in $i$-th position by $x_1$ and anything
 else by 0. The matrices defining these substitutions give the desired
 abp-reduction, which we explain now.

As in Theorem \ref{thm:vpcomplete}, we convert this automaton into a
ABP $A$ computing the homogeneous polynomial of degree $2n$. We now
describe matrices we substitute for each parenthesis. Let $M_p$ be the
matrix we substitute for parenthesis $p$ whose rows and columns are
labelled by nodes of this ABP A.

We define the matrix $M_p$ for parenthesis $p$ as follows:


\begin{displaymath}
m_{u,v}=M_p[u,v]= \left\{ 
\begin{array}{ll} x_0 & \textrm{if $p \in
    \{(,)\}$ and $\exists e=(u,v) \in$ E(A) and label of $e$ is in $
    \{(,)\}$ }\\ x_1 & \textrm{if $p \in \{[,]\}$ and $\exists e=(u,v)
    \in$ E(A) and label of $e$ is in $ \{[,]\}$ }
\end{array} 
\right.
\end{displaymath}
where E(A) is the edge set of the automaton A.

We now turn to the converse problem. In fact, we only need to observe
that $\PAL^2$ is also $\leq_{abp}$-reducible to $D_2$, where $\PAL^2$
is the square of the Palindrome
polynomial. I.e.\ $\PAL^2=(\PAL_n\PAL_n)_{n\geq0}$. We can easily
reduce $\PAL_n\PAL_n$ to $D_{2,2n}$ by repeating the automaton
construction giving $\PAL_n$ from $D_{2,n}$ twice. The automaton will
zero out all monomials of $D_{2,2n}$ except those of the form
$u_1.u_2$ where $u_1$ has an equal number of $($ and $)$ and equal
number of [ and ] and similarly $u_2$.

Furthermore, while reading $u_1$ the automaton will do exactly as the
reductions of $\PAL_n$ to $D_{2,n}$ and also for $u_2$ the same. This
will yield the polynomial $\PAL_n\PAL_n$. Hence $\PAL^2
\leq_{abp}D_2$. However, $\PAL^2\nleq_{abp}\PAL$ because, as shown in
\cite {LMS15}, skew circuits computing $\PAL^2$ require exponential
size. This completes the proof sketch.
\end{proof}

\subsection{Dyck depth hierarchy inside $\vpnc$}

We now show that the nesting depth of Dyck polynomials can be used to
obtain a strict hierarchy of p-families within $\vpnc$. This
hierarchy roughly corresponds to the $\vncnc$ hierarchy.

\begin{definition}
A p-family $f=(f_n)$ is in $\vncnc^i$ if there is a family of circuits
$(C_n)$ for $f$ such that each $C_n$ is of polynomial size and degree,
and is of $\log^in$ depth.
\end{definition}

The classes $\vncnc^i, i=1,2,\ldots$ are contained in
$\vpnc$. Furthermore, it is easy to show using Nisan's rank argument
that $\vncnc^i, i=1,2,\ldots$ form a strict
hierarchy.\footnote{Palindromes of length $\log^{i+1}n$ have circuits
  of depth $\log^{i+1} n$ and polynomial in $\log^{i+1}n$
  size. However, circuits of depth $\log^{i} n$ for it require
  superpolynomial size.}

It turns out that Dyck polynomials of nesting depth $\log^{i+1} n$ are
hard for $\vncnc^i$ w.r.t.\ $\leq_{abp}$ reductions. Indeed, this follows
from inspection of the proof of Theorem~\ref{thm:vpcomplete}.

\begin{definition}[Nesting depth]
The nesting depth of a string in $D_2$ is defined as follows:
\begin{itemize}
 \item $()$ and $[]$ have depth 1.
 \item If $u_1$ has depth $d_1$ and $u_2$ has depth $d_2$, $u_1u_2$
   has depth $max\{x_1,d_2\}$ and $(u_1),[u_1]$ have depth $d_1+1$.
 \end{itemize} 
\end{definition}

Let $W^{(k)}_{2,n}$ denote the set of all monomials in $D_{2,n}$ of
depth at most $k$ and degree $2n$. We define the polynomial
$D^{(k)}_{2,n}=\sum_{u \in W^{(k)}_{2,n}}u$ and denote the
corresponding p-family as $D^{(k)}_{2}$. In this definition we allow
$k$ to be a function $k(n)$ of $n$, where
$D^{(k)}_{2}=(D^{(k)}_{2,n})_{n\geq0}$.

\begin{theorem}
 Let $k_1=\omega(\log n)$ and $k_2(n)\geq \omega(k_1(n))$ for all $n$. Then $D_2^{k_2}\nleq_{abp}D_2^{k_1}$ but $D_2^{k_1}\leq_{abp}D_2^{k_2}$.
\end{theorem}

\begin{proof}
 Suppose $D_2^{(k_2)} \leq_{abp} D_2^{(k_1)}$. Then there are
 polynomials $p(n)$ and $q(n)$ such that there is a matrix
 substitution $\phi_n$ for the variables $X$ of $D_{2,p(n)}^{(k_1)}$
 with the property that
\[ 
D_{2,p(n)}^{(k_1)}(\phi_n(X))(1,q(n))=D_{2,n}^{(k_2)}, 
\]
where $\phi_n$ is a $q(n) \times q(n)$ matrix substitution for each
variable in $X$. Now, the polynomial $D_{2,p(n)}^{(k_1(n))}$ has an
ABP of size $2^{k_1(n)}.poly(n)$ (this ABP can be constructed by
keeping the stack content as part of the DFA state for stack size at
most $k_1(n)$). Combined with the matrix substitutions $\phi_n$, we
obtain a $2^{k_1(n)}.\poly(n)$ size ABP for the polynomial
$D_{2,n}^{(k_2(n))}$.

Furthermore, the reduction from $\PAL$ to $D_2$
(Theorem~\ref{pal-to-d2}) can be easily modified to show that
$\PAL_{k_2}\leq D_{2,n}^{(k_2(n))}$ (the reduction will work only with
the prefixes of length $2k_2(n)$ of $D_{2,n}^{(k_2(n))}$ and
substitute rest by 1, if the prefix has same number of left and right
brackets and 0 otherwise).

But by Nisan's \cite{N91} rank argument $\PAL_{k_2}$ requires
$2^{\Omega(k_2)}$ size ABPs contradicting the above
$2^{k_1(n)}.poly(n)$ size ABP.

We now show the reduction $D_2^{k_1}\leq_{abp}D_2^{k_2}$. We design a
DFA with $O(n.k_1(n))$ states that takes strings $u$ of length $2n$
over $\{(,),[,]\}$ with an equal number of ( \& ) and an equal number
of [ \& ] s.t in every prefix $s$ of $u$, the number of left brackets
exceed the number of right brackets by at most $k_1(n)$.

Corresponding to this DFA we can create matrix substitutions which
replace each variable $x \in \{(,),[,]\}$ by itself if the string is
accepted and otherwise, the $(2n)$-th variable by 0. Let $\phi_n$
define this matrix substitution. Then
$D_{2,n}^{k_2(n)}\phi_n(X))=D_{2,n}^{k_1(n)}$, where
$X=Var(D_{2,n}^{k_2(n)})$. This completes the proof.
\end{proof}

\section{More on $\vnpnc$-Completeness}
\label{sec:more_vnpc}
Apart from the polynomial family $\PER_d$, we know from \cite{AS10} that the polynomial family $\DET_d$ is $\vnpnc$-complete for $\abp$-reductions. 
In this section we show some new $\vnpnc$-complete p-families
w.r.t.\ $\le_{abp}$ reductions and raise some open questions. In
Theorem~\ref{thm:notcomplete} we saw that $ID$ is not
$\vnpnc$-complete w.r.t.\ $\le_{iproj}$ reductions. However, we do not
know if $ID$ is $\vnpnc$-complete w.r.t.\ $\le_{abp}$ reductions.

Motivated by this question we consider a generalized version of
$ID$ which we call $ID^*$ defined as follows:

For each positive integer $n$, let $W_n$ denote the set of all degree
$n$ monomials of the form $x_{1,i_1}\ldots x_{n,i_n}$, over the
variable set $\{x_{ij}\mid 1\le i,j\le n\}$.
\[
ID^*_n=\sum_{w \in W_n}\underbrace{ww\dots w}_{n^2-times}.
\]

\begin{theorem}
\label{thm:more_id*}
 $\PER \leq_{abp} ID^*_n$. 
\end{theorem}

\begin{proof}
Consider the permanent polynomial $\PER_n$ defined on the variable set
$V_n=\{x_{ij}\mid 1\le i,j\le n\}$. We design a polynomial in $n$
sized deterministic automaton A with the following properties:

\begin{enumerate}
\item It takes inputs $w_1w_2\dots w_{n^2}$ over alphabet $V_n$,
where each $w_i$ is of length $n$.  

\item It checks that each $w_i$ is a monomial of the form
  $w=X_{1i_1}\ldots X_{ni_n}$. I.e.\ the first index of the variables
  is strictly increasing from $1$ to $n$.

\item For the $i^{th}$ block $w_i$, since $1\le i\le n^2$, we can
  consider the index $i$ as a pair $(j,k), 1\le j,k\le n$. While
  reading the $i^{th}$ block $w_i=X_{1i_1}\ldots X_{ni_n}$ the
  automaton checks that $i_j \neq i_k$ if $j\ne k$.
\end{enumerate}

The automaton A can be easily realized as a DAG with $n^3$ layers. The
first layer has the start state $s$ and the last layer has one
accepting state $t$ and one rejecting state $t'$. Transitions are only
between adjacent layers, from $i$ to $i+1$ for each $i$. Layers are
grouped into blocks of size $n$. Let the blocks be
$B_1,B_2,\ldots,B_{n^2}$. In block $B_i$, the transitions of the
automaton will check if $i_j\ne i_k$ assuming $j\ne k$, where
$i=(j,k)$. The automaton can have the indices $j$ and $k$ hardwired in
the states corresponding to block $B_i$ and easily check this
condition. If for any block $B_i$, the indices $i_j=i_k$ then the
automaton stores this information in its state and in the end makes a
transition to the rejecting state $t'$.

Finally, the matrices of the automaton have to effect substitutions in
order to convert monomials of $P$ into monomials of $\PER$. The
matrices will replace $x_{ij}$ by the same variable $x_{ij}$ in the
first block $B_1$ and by $1$ in all subsequent blocks. The polynomial
$ID^*_n$ when evaluated on these matrices will have the permanent
polynomial $\PER_n$ in the $(s,t)^{th}$ entry of the resulting matrix.
This completes the proof of the theorem.
\end{proof}

Let $\chi:S_n\to \F\setminus\{0\}$ be any polynomial-time computable
function assigning nonzero values to each permutation in $S_n$. We
define a generalized permanent
\[
\PER^\chi_n = \sum_{\sigma\in
  S_n}\chi(\sigma)x_{1\sigma(1)}x_{2\sigma(2)}\dots x_{n\sigma(n)}.
\]

Clearly $\PER^{\chi}=(\PER_n^\chi)$ is a p-family that is in
$\vnpnc$. For which functions $\chi$ is $\PER^\chi$ $\vnpnc$-complete?
In other words, does the hardness of the noncommutative permanent
depend only on the nonzero monomial set (and the coefficients are not
important)? We give a partial answer to this question. Define
 
\begin{align*}
 \PER^*=\sum_{\sigma \in
   S_n}\underbrace{\overline{X}_{\sigma}\overline{X}_{\sigma}\ldots
   \overline{X}_{\sigma}}_{n-times}, \text{ where
   $\overline{X}_{\sigma}$ is the monomial $x_{1\sigma(1)}\ldots x_{n\sigma(n)}$}.
\end{align*}

\begin{proposition}
$\PER^*$ is  $\vnpnc$-complete.
\end{proposition}

The above proposition is easy to prove: $\PER^*$ is in $\vnpnc$
because coefficients of each monomial is polynomial-time computable
from the monomial \cite{HWY10b}. Furthermore, $\PER$ is
$\le_{iproj}$-reducible to $\PER^*$ by substituting $1$ for all except
the first $n$ variables in every monomial. 

Now, consider the polynomial
\[
 \PER^{*,\chi}=\sum_{\sigma \in
   S_n}\chi(\sigma)\underbrace{\overline{X}_{\sigma}\overline{X}_{\sigma}\ldots
   \overline{X}_{\sigma}}_{n-times}.
\]

We prove the following theorem about $\PER^\chi$ and $\PER^{*,\chi}$
under assumptions about the function $\chi$.

\begin{theorem}\label{thm:vnpc_gen}
 Suppose the function $\chi$ is such that $|\chi(S_n)|\le p(n)$ for
 some polynomial $p(n)$ and each $n$. Then
\begin{itemize}
\item If $\chi$ is computable by a $1$-way logspace Turing machine
  then $\PER\leq_{abp}\PER^{\chi}$.
\item If $\chi$ is computable by a logspace Turing machine then
  $\PER\leq_{abp}\PER^{*,\chi}$.
\end{itemize}
\end{theorem}

\begin{proof}
We explain the second part of the theorem. The first part follows from
the proof of the second. The idea is to construct an automaton from
the given logspace machine such that for a given $\sigma \in S_n$, the
automaton computes $\frac{1}{\chi(\sigma)}$ in the field $\F$.

Let $T$ be a logspace Turing machine which uses space $s=O(\log n)$,
computing $\chi$. Thus, total running time of $T$ is bounded by
$P(n)$, where $P(n)$ is some fixed polynomial in $n$. Since the range
of $\chi$ is $p(n)$ bounded in size, we can encode in a state of
the automaton the following:

\begin{itemize}
 \item Input head position,
 \item Content of working tape, and
 \item Content of output tape.
\end{itemize}

The number of states is bounded by a polynomial in $n$.  We can
convert this log-space machine $T$ on input $\sigma$ into a one-way
log-space machine $T'$ on a modified input as follows:

\begin{itemize}
 \item The input to $T'$ is the concatenation of $P(n)$ copies of
   $\sigma$. Thus the input to $T'$ is of the form
   $\sigma\sigma\ldots\sigma$, with $P(n)$ many $\sigma$.
 \item At a step $i$, $T'$ reads from the $i^{th}$ copy.
\end{itemize}

The difference between machine $T'$ and $T$ is that $T'$ is a
\emph{1-way} logspace machine whose input head moves always to the
right. For $\sigma\in S_n$, we can convert $T'$ into a deterministic
automaton with $\poly(n)$ many states as follows: there are only
polynomially many instantaneous descriptions of $T'$. This consists of
the input head position, the work tape contents and head position, and
the current output string (which is a prefix of some element in the
range $\chi(S_n)$). When this automaton completes reading the input,
suppose the state $q$ contains the output element
$\alpha=\chi(\sigma)$. The automaton has a transition from $q$ to the
unique final state $t$ labeled by scalar $1/\chi(\sigma)$.

Finally, we can modify this automaton to work on the monomials
$\overline{X}_\sigma\overline{X}_\sigma\dots\overline{X}_\sigma$,
where it replaces all but the first block of variables by $1$.

When the polynomial $\PER^{*,\chi}$ is evaluated on the matrices
corresponding to the above automaton (with the substitutions), the
$(s,t)^{th}$ entry of the output matrix will be the permanent
polynomial $\PER_n$.
\end{proof}



\begin{thebibliography}{HWY10b}

\bibitem[AJS09]{AJS09}
Vikraman Arvind, Pushkar~S. Joglekar, and Srikanth Srinivasan, \emph{Arithmetic
  circuits and the hadamard product of polynomials}, {IARCS} Annual Conference
  on Foundations of Software Technology and Theoretical Computer Science,
  {FSTTCS} 2009, December 15-17, 2009, {IIT} Kanpur, India, 2009, pp.~25--36.

\bibitem[AS10]{AS10}
Vikraman Arvind and Srikanth Srinivasan, \emph{On the hardness of the
  noncommutative determinant}, Proceedings of the 42nd {ACM} Symposium on
  Theory of Computing, {STOC} 2010, Cambridge, Massachusetts, USA, 5-8 June
  2010, 2010, pp.~677--686.

\bibitem[BBB00]{Beim}
Amos Beimel, Francesco Bergadano, Nader~H. Bshouty, Eyal Kushilevitz, and
  Stefano Varricchio, \emph{Learning functions represented as multiplicity
  automata}, Journal of the ACM \textbf{47} (2000), no.~3, 506--530.

\bibitem[BR11]{berstel}
J.~Berstel and C.~Reutenauer, \emph{Noncommutative rational series with
  applications}, Encyclopedia of Mathematics and its Applications, Cambridge
  University Press, 2011.

\bibitem[B{\"{u}}r99]{Burg99}
Peter B{\"{u}}rgisser, \emph{On the structure of valiant's complexity classes},
  Discrete Mathematics {\&} Theoretical Computer Science \textbf{3} (1999),
  no.~3, 73--94.

\bibitem[CS63]{cs63}
Noam Chomsky and Marcel~Paul Sch\"utzenberger, \emph{{The Algebraic Theory of
  Context-Free Languages}}, Computer Programming and Formal Systems
  (P.~Braffort and D.~Hirshberg, eds.), Studies in Logic, North-Holland
  Publishing, 1963, pp.~118--161.

\bibitem[DSW94]{Davis94}
Martin~D. Davis, Ron Sigal, and Elaine~J. Weyuker, \emph{Computability,
  complexity, and languages (2nd ed.): Fundamentals of theoretical computer
  science}, Academic Press Professional, Inc., San Diego, CA, USA, 1994.

\bibitem[HWY10a]{HWY10a}
Pavel Hrubes, Avi Wigderson, and Amir Yehudayoff, \emph{Non-commutative
  circuits and the sum-of-squares problem}, Proceedings of the 42nd {ACM}
  Symposium on Theory of Computing, {STOC} 2010, Cambridge, Massachusetts, USA,
  5-8 June 2010, 2010, pp.~667--676.

\bibitem[HWY10b]{HWY10b}
Pavel Hrubes, Avi Wigderson, and Amir Yehudayoff, \emph{Relationless completeness and separations}, Proceedings of the
  25th Annual {IEEE} Conference on Computational Complexity, {CCC} 2010,
  Cambridge, Massachusetts, June 9-12, 2010, 2010, pp.~280--290.

\bibitem[Lad75]{Ladner75}
Richard~E. Ladner, \emph{On the structure of polynomial time reducibility}, J.
  {ACM} \textbf{22} (1975), no.~1, 155--171.

\bibitem[LMS15]{LMS15}
Nutan Limaye, Guillaume Malod, and Srikanth Srinivasan, \emph{Lower bounds for
  non-commutative skew circuits}, Electronic Colloquium on Computational
  Complexity {(ECCC)} \textbf{22} (2015), 22.

\bibitem[Nis91]{N91}
Noam Nisan, \emph{Lower bounds for non-commutative computation (extended
  abstract)}, STOC, 1991, pp.~410--418.

\bibitem[RS05]{raz05PIT}
Ran Raz and Amir Shpilka, \emph{Deterministic polynomial identity testing in
  non-commutative models}, Computational Complexity \textbf{14} (2005), no.~1,
  1--19.

\bibitem[Str69]{str}
Volker Strassen, \emph{Gaussian elimination is not optimal}, Numerische
  Mathematik \textbf{13} (1969), no.~4, 354--356.

\bibitem[Val79]{Val79}
Leslie~G. Valiant, \emph{Completeness classes in algebra}, Proceedings of the
  11h Annual {ACM} Symposium on Theory of Computing, April 30 - May 2, 1979,
  Atlanta, Georgia, {USA}, 1979, pp.~249--261.

\end{thebibliography}

\end{document}